\documentclass[runningheads]{llncs}
\usepackage{float}
\usepackage{graphicx}
\usepackage{tikz}
\usepackage{todonotes}
\usepackage{amsmath}
\usepackage{amsfonts}
\usepackage{amssymb}
\usetikzlibrary{decorations.pathreplacing}

\usepackage{geometry}
\newcommand{\cO}{\mathcal{O}}

\newcommand{\convergensen}{\textsc{Convergence}}

\newcommand{\disj}{\textsc{Disjointness}}

\newcommand{\electionpredi}
{\textsc{Election-Prediction}}

\newcommand{\inp}{\textsf{in}}
\newcommand{\out}{\textsf{out}}
\newcommand{\id}{\textsf{id}}
\newcommand{\In}{\textsf{In}}

\begin{document}
\title{The Hardness of 
Local Certification of Finite-State Dynamics}
%

\author{Diego Maldonado\inst{1}\and
Pedro Montealegre\inst{2}\and
Martín Ríos-Wilson\inst{2}}

\authorrunning{D. Maldonado, P. Montealegre, M. Ríos-Wilson}

\institute{Facultad de Ingeniería, Universidad Católica de la Santísima Concepción, Chile \email{dmaldonado@ucsc.cl}\and
 Facultad de Ingeniería y Ciencias, Universidad Adolfo Ibáñez, Chile
\email{\{p.montealegre,martin.rios\}@uai.cl}}

\maketitle              

\begin{abstract}

Finite-State Dynamics (FSD) is one of the simplest and constrained distributed systems. An FSD is defined by an $n$-node network, with each node maintaining an internal state selected from a finite set. At each time-step, these nodes synchronously update their internal states based solely on the states of their neighboring nodes.

Rather than focusing on specific types of local functions, in this article, our primary focus is on the problem of determining the maximum time required for an FSD to reach a stable global state. This global state can be seen as the acceptance state or as the output of a distributed computation. For fixed $k$ and $q$, we define the problem $\text{convergence}(k,q)$, which consists of deciding if a $q$-state FSD converges in at most $k$ time-steps.

Our main focus is to study the problem $\text{convergence}$ from the perspective of distributed certification, with a focus on the model of proof-labeling schemes (PLS). First, we study the problem $\text{convergence}$ on arbitrary graphs and show that every PLS has certificates of size $\Theta(n^2)$ (up to logarithmic factors). Then, we turn to the restriction of the problem on graphs of maximum degree $\Delta$. Roughly, we show that the problem admits a PLS with certificates of size $\Delta^{k+1}$, while every PLS requires certificates of size at least $2^{k/6} \cdot 6/k$ on graphs of maximum degree \(3\).

\keywords{Local Certification \and Proof Labeling Schemes \and Finite State Dynamics}
\end{abstract}

\section{Introduction}

Networks serve as the backbone of numerous scientific domains, ranging from the social sciences, where they represent human connections, to logistics, as seen in traffic patterns, and even electrical engineering, as in circuitry. Distributed computing explores the capabilities and constraints of algorithms that operate across these networks. Given the pervasive influence of the Internet, most contemporary models aim to understand devices capable of accessing it. Typically, distributed computing employs the message-passing model, where nodes can send extensive messages to neighboring nodes and perform local computations. However, some networks, such as emerging wireless networks like ad-hoc or sensor networks, do not perfectly conform to this traditional model. Their underlying devices have constraints that do not align with the classical message-passing model

A recent trend involves using distributed computing techniques, especially the message-passing model, for sub-microprocessor networks like those in biological cells or nano-mechanical devices \cite{emek2013stone}. However, the fundamental differences in capability between biological or nano nodes and silicon-based devices necessitate a distinct network model, one designed for nodes inherently more limited than Turing machines. A natural model to consider in that context is the one of finite-state dynamics \cite{costa2023effective,frischknecht2013convergence,guseo2009modelling}. Within this framework, numerous agents are interconnected through an undirected graph known as the \emph{interaction graph}. Each agent communicates solely with its immediate neighbors. Every node on this graph has an inherent local state, with the collective states of all nodes forming a \emph{configuration} or global state for the dynamic system. This internal state gets updated, relying on its prior state and those of its neighbors. This updating is governed by a \emph{local rule}. Moreover, the function that transforms one configuration into another using local rules is termed the \emph{global rule}. As this rule is repeatedly applied, a series of global configurations ensues. Given that there is a finite number of configurations, the dynamics will eventually lead to a repetitive sequence of node states. If, after successive applications of the global rule, the system reaches a stable configuration where the internal state of each node remains unchanged, we say that the system \emph{converges}.

An interesting feature of the finite-state dynamics formalism is its capacity to describe complex behavior while utilizing a constant-size information representation through the states of the nodes \cite{atlan2000self}. In the realm of theoretical computer science, the study of finite-state dynamics revolves around several key topics. One significant question involves investigating whether the structure of the underlying interaction graph can influence its dynamics \cite{gadouleau2017stability,gadouleau2020influence,gadouleau2016simple}. Another well-explored approach is the study of decision problems linked to the dynamics, often motivated by specific questions pertaining to particular models. In this context, a central question centers on determining the point at which the complexity of the dynamics, viewed as the complexity of a particular decision problem, becomes related to the dynamical behavior.

In the field of distributed systems, numerous models exhibit similarities or share characteristics with finite-state dynamics. An example of a model in which the messages distributed by the nodes in the network are bounded is the case of the beeping model \cite{cornejo2010deploying,flury2010slotted}. In this model, nodes can either beep or stay silent, and they can only distinguish between two situations: when all their neighbors are silent or when at least one neighbor is beeping. In addition, they can either beep or listen to their neighbors, but they cannot do both at the same time. A crucial difference between the FSD model and the beeping model is that in the latter, the local computation is performed by a Turing machine (so the set of states is not finite). This model has interesting applications as the beeping process can be interpreted as communication between simple machines, such as sensors, or it can be seen as a way to communicate messages in a biological network \cite{afek2013beeping,afek2011biological}.

Another example of a model related to finite-state dynamics is the model called networked state machines (nFSM) introduced in \cite{emek2013stone}. In this case, the local computation is done using a finite set of states, but the essential difference between this model and finite-state dynamics is that nFSM operates asynchronously and incorporates elements of randomness. In the previously cited papers, the dynamics induced by these models are interpreted as the distributed computing of coloring or a maximal independent set \cite{afek2013beeping,afek2011biological,emek2013stone}.

Generally speaking, the main measure of performance in distributed systems is the number of rounds required to output a solution in the worst case. In this sense, studying the time required for the system to compute a solution, the convergence time, is crucial. However, to our knowledge, this has not been addressed from a more dynamics-centered approach or a local certification approach.

Existing literature has established a connection between the computational complexity of these problems and the dynamical properties of the system \cite{rios2023intrinsic,barrett2006complexity,goles2014computational}. Nevertheless, to our knowledge, these problems have not been explored from a distributed standpoint.

Recently, a novel approach based on distributed algorithms and local decisions has been introduced \cite{maldonado2023local}. This approach has been applied to study opinion dynamics in social systems. In this work, we extend this approach to examine a broader framework, specifically focusing on distributed certification for determining whether the dynamics converge.

\subsection{The model}

We now define the model more formally, while we refer to the preliminaries section for further details.  We consider simple finite undirected graphs \(G\). Each node has one over a finite set of states \(Q\). A node \(u\) of \(G\) also knows  a deterministic local function \(f_u\), that specifies how the node updates its state, computed according to the states on its closed neighborhood. In order to simplify the model, we consider only synchronous dynamics, meaning that all nodes update their states at the same time on each time-step. The synchronous application of the local functions on every node produces a \emph{global function} \(F\) over \(G\), which establishes how the configuration of states of every node evolves over time. The pair \((G,F)\) is called a \emph{finite-state dynamic system}. 

Since the possible state configurations of the system is finite, and the local functions are deterministic, we have that every sequence of state configurations of a finite-state dynamics system is eventually periodic. A state configuration is called a \emph{fixed point} if the dynamics does not change under the application of the global function. A finite-state dynamics system is \emph{convergent} if every state configuration reaches a fixed point after a certain number of times-steps. The \emph{convergence time} of a system is the maximum number of time-steps required to reach a fixed point from any initial configuration. Obviously, the \emph{convergence time} of a system is finite only if the system is convergent.

\subsection{The \convergensen\ problem}

In the study of finite-state dynamics, there are multiple possible research questions. Most of these questions are related to the study of one particular dynamical behavior. However, instead of focusing on a specific system, the main aim of this work is to focus on identifying properties of an arbitrary dynamics that can be studied using distributed methods. A natural task in this context consists of predicting the long-term behavior of the system. For example, it might be interesting to ask if observing the dynamics for a given number of time-steps is enough for the system to attain a particular state (such as an acceptance state, for instance), or ascertain whether a computation processed by the network will eventually converge.

Within distributed systems characterized by finite states, the properties based on localized knowledge of the update rules constitute a key element to consider. In fact, given the inherent nature of distributed systems, it becomes crucial to understand the evolution of the dynamics and the long-term behavior of the system based solely on this local data.

The main research question of this article is: Using local knowledge, can we predict if a dynamic system will converge within a given number of time-steps? More precisely, given a pair of positive integers $k$ and $q$, we define the problem $\convergensen(k,q)$ which consists of, given a $q$-state finite-state dynamic system, determining if the dynamics converges in at most $k$ time-steps.

This problem is not solvable by a local algorithm, as the state evolution of a node might be affected by the state configuration in remote locations of the network. For that reason, we tackle the $\convergensen(k,q)$ problem from the perspective of local certification, more specifically, \emph{proof-labeling schemes}.

\subsection{Local certification.} A local certification algorithm for a distributed decision problem is a prover-verifier pair where the prover is an untrustworthy oracle assigning certificates to the nodes, and the verifier is a distributed algorithm enabling the nodes to check the correctness of the certificates by a certain number of communication rounds with their neighbors. Note that the certificates may not depend on the instance \( G \) only, but also on the identifiers \( \id \) assigned to the nodes. Proof-labeling schemes (PLSs) are a specific type of local certification algorithms, where the information exchanged between the nodes during the verification phase is limited to only one round, and the contents are limited to the certificates. The prover-verifier pair must satisfy the following two properties.

 \medskip

\noindent{\it Completeness:} On yes-instances, the untrustworthy prover can assign certificates to the nodes such that the verifier accepts at all nodes;

\medskip

\noindent{\it Soundness:} On no-instances, for every certificate assignment to the nodes by the untrustworthy prover, the verifier rejects in at least one node.

\medskip

The main complexity measure for proof-labeling schemes is the size of the certificates assigned to the nodes by the prover.

\subsection{Our Results}

First, we observe that for each pair of positive integers \(k\) and \(q\), the \(\convergensen(k, q)\) problem can be verified on \(n\)-node graphs with certificates of size \[b_{\max}(k+1) \cdot \log(\id_{\max}) \cdot |f_{\max}(n,q)|,\] where \(b_{\max}(k+1)\) is the maximum number of edges having an endpoint in a node at distance at most \(k+1\), the identifiers are assigned with values in \([\id_{\max}]\) and \(|f_{\max}(n,q)|\) is the maximum number of bits required to encode a local function. This is achieved by providing each node with the necessary part of the network to simulate \(k+1\) time-steps of any configuration. In our first main result, we show that such a simple algorithm is in fact the best one possible in general, up to logarithmic factors. More precisely, we show that every Proof-Labeling Scheme (PLS) for \(\convergensen(2, 4)\) requires certificates of size \(\Omega(n^2/(\log n))\).

Observe that even if we consider a finite number of states, the encoding of a local function can be exponentially large relative to the number of neighbors of a node. Therefore, we consider local rules that admit \emph{succinct representations}, which can be encoded using a logarithmic number of bits per neighbor. Examples of such local rules are the ones found in artificial neural networks, the modeling of opinion dynamics, or different biological processes. Our lower-bound does not hide the complexity in the encoding of the local rules, as it holds even when the problem is restricted to local functions that admit a succinct representation.

Later, we focus on graphs with bounded degree. Restricted to this case, the straightforward approach gives a PLS for \(\convergensen(k, q)\) with certificates of size at most \[\Delta((\Delta -1)^{k+1}-1)\cdot q^{\Delta+1} \log (q) \cdot \log (\id_{\max}) \] where \(\Delta\geq 2\) is the maximum degree of the input graph. In our second main result, we show that the exponential dependency on \(k\) cannot be avoided in the bounded degree case. More precisely, we show that every PLS for problem \(\convergensen(k, 3)\) requires certificates of size at least \(2^{k/6} \cdot 6/k\) even when the problem is restricted to input graphs of maximum degree \(3\).

\subsection{Related Work}

\noindent {\bf Finite-state dynamics and majority voting dynamics.} An interesting distributed problem is the one of predicting the result of a two-candidate election. This process is typically modeled as a majority voting dynamics. The model can be seen as a group of $n$ agents represented by nodes in a network who are surveyed about their preferred candidate in an upcoming election with two choices. Over a series of time steps, each agent adjusts their vote based on the majority opinion of their network neighbors, ultimately determining the leading candidate after $T$ time steps. The challenge lies in predicting the leading candidate which in general is a very hard task.

In \cite{maldonado2023local} the authors study the problem $\electionpredi$, consisting of predicting the leading candidate after a certain number of time-steps from the perspective of local certification. In particular, they show that graphs with sub-exponential growth admit a proof labeling scheme of size $\mathcal{O}(\log n)$ for problem $\electionpredi$. Additionally, they deduce upper bounds for graphs with bounded degree, where certificate sizes are sub-linear in $n$. Furthermore, they explore lower bounds for the unbounded degree case, establishing that the local certification of $\electionpredi$ on arbitrary $n$-node graphs require certificates of at least $\Omega(n)$ bits. Interestingly, the authors show that the upper bounds are tight, even for graphs with constant growth.

In terms of the techniques used in this paper, the authors present an interesting approach for deriving an upper bound based on the analysis of the maximum number of time steps in which an individual may change their opinion during the majority dynamics. In particular, they show for different families of graphs (one of them being the graphs with sub-exponential growth) that this number is bounded when the dynamics are observed at every two time steps. In addition, the lower bounds are deduced via a reduction to the disjointedness problem in non-deterministic communication complexity.\\

\noindent {\bf Local certification.} 
Since the introduction of PLSs~\cite{KormanKP10}, various variants have been introduced. As we mentioned, a stronger form of PLS are locally checkable proofs~\cite{goos2016locally}, where each node can send not only its certificates but also its state and look around within a given radius. Other stronger forms of local certifications are \(d\)-PLS~\cite{FeuilloleyFHPP21}, where nodes perform communication at a distance of \(d \geq 1\) before deciding. Authors have studied many other variants of PLSs, such as randomized PLSs \cite{fraigniaud2019randomized}, quantum PLSs~\cite{FraigniaudGNP21}, interactive protocols~\cite{CrescenziFP19,kol2018interactive,NaorPY20}, zero-knowledge distributed certification~\cite{BickKO22}, and certain PLSs that use global certificates in addition to the local ones~\cite{FeuilloleyH18}, among others. On the other hand, some trade-offs between the size of the certificates and the number of rounds of the verification protocol have been exhibited~\cite{FeuilloleyFHPP21}. Also, several hierarchies of certification mechanisms have been introduced, including games between a prover and a disprover~\cite{BalliuDFO18,FeuilloleyFH21}.

PLSs have been shown to be effective for recognizing many graph classes. For example, there are compact PLSs (i.e., with logarithmic size certificates) for the recognition of acyclic graphs \cite{KormanKP10}, planar graphs~\cite{feuilloley2020compact}, graphs with bounded genus~\cite{EsperetL22}, and \(H\)-minor-free graphs, provided that \(H\) contains no more than four vertices~\cite{BousquetFP21}. In a recent breakthrough, Bousquet et al.~\cite{bousquet2021local} proved a 'meta-theorem', stating that there exists a PLS for deciding any monadic second-order logic property with \(O(\log n)\)-bit certificates on graphs of bounded \emph{tree-depth}. This result has been extended by Fraigniaud et al~\cite{FraigniaudMRT22} to the larger class of graphs with bounded \emph{tree-width}, using certificates on \(O(\log^2 n)\) bits.

\section{Preliminaries}

In this article, we denote by \([m,n]\) the set of integers greater than or equal to \(m\) and less than or equal to \(n\). We also denote \([n]\) as the interval \([1,n]\).

Let \(G = (V,E)\) be a graph. We denote by \(N_G(v)\) the \emph{neighborhood of \(v\) in \(G\)}, defined by \(N_G(v) = \{u \in V: \{u,v\} \in E\}\). The \emph{degree} of \(v\), denoted \(d_G(v)\), is the cardinality of \(N_G(v)\). The \emph{maximum degree} of \(G\), denoted \(\Delta_G\), is the maximum value of \(d_G(v)\) taken over all \(v \in V\). We denote by \(N_G[v]\) the set \(N_G(v) \cup \{v\}\) and call it the \emph{close neighborhood} of \(v\). We say that two nodes \(u, v \in V\) are \emph{connected} if there exists a path in \(G\) joining them. In the following, we only consider connected graphs. The \emph{distance} between \(u, v\), denoted \(d_G(u,v)\), is the minimum length (number of edges) of a path connecting them. For a graph \(G\), \(v \in V(G)\), and \(p \geq 0\), we denote \(B_v(p)\) as the set of all edges where one of the endpoints is a node at distance at most \(p\) from \(v\) in \(G\). We denote by \(b_v(p)\) the cardinality of \(B_v(p)\), and denote by \(b_{\max}(p)\) the maximum \(b_v(p)\) over \(v \in V(G)\). In the following, we omit the sub-indices when they are obvious by the context.


\subsection{Finite state dynamics} Let \(G=(V,E)\) be a graph, \(Q\) a finite set. A \emph{finite state dynamic system over \(G\)} is a function \(F: Q^V \to Q^V\) such that, for each \(v \in V\) there exists a function \(f_v : Q^{N[v]}\rightarrow Q\) satisfying \(F(x)_v = f_v(x|_{N[v]})\). The functions \(\{f_v\}_{v\in V}\) are called the \emph{local functions} of \(F\) and the elements of \(Q\) are the \emph{states} of \(F\). The elements of \(Q^V\) are called \emph{configurations} of \(G\).

We consider the model where \(F\) is distributed over the network in a way that each node \(v\) receives its local function \(f_v: Q^{N[v]} \rightarrow Q\) as input. A vertex \(v\) identifies \(d(v)\) ports that enumerate its incident edges, where \(d(v)\) is the size of the neighborhood of \(v\). Each input variable of \(f_v\) is identified with one of the ports, except one that is identified with \(v\). Therefore, \(f_v\) can be encoded in at most \(|Q|^{d(v)+1}\log {|Q|} + (d(v)+1)\log(d(v)+1)\) bits.

In some cases we are interested in the local functions that can be encoded with much fewer bits. For instance, consider the dynamics with states \(Q = \{-1,1\}\) and where each node takes the majority state over its neighbors, which can be encoded with a constant number of bits per node. We say that \(F: Q^V \to Q^V\) is \emph{succinct} if, for each \(v \in V\) we have that \(f_v\) can be encoded with \(\mathcal{O}(d(v)\log (d(v)))\) bits.

\subsection{Local decision} Let \(G =(V,E)\) be a simple connected \(n\)-node graph. A \emph{distributed language} \(\mathcal{L}\) is a Turing-decidable collection of tuples \((G,\id, \In)\), called \emph{network configurations}, where \(\In: V\rightarrow \{0,1\}^*\) is called an \emph{input function} and \(\id: V \rightarrow [n^c]\) is an injective function that assigns to each vertex a unique identifier in \([n^c]\) with \(c>1\). In this article, all our distributed languages are independent of the \(\id\) assignments. In other words, if \((G,\id, \In) \in \mathcal{L}\) for some \(\id\), then \((G,\id', \In) \in \mathcal{L}\) for every other \(\id'\).

Given \(d>0\), a \emph{local decision algorithm} for a distributed language \(\mathcal{L}\) is an algorithm on instance \((G, \id, \In)\), where each node \(v\) in \(V(G)\) receives the subgraph induced by all nodes within a distance of at most \(d\) from \(v\), including their identifiers and inputs. The integer \(d>0\) depends only on the algorithm, not on the size of the input. Each node performs unbounded computation on the information received, and decides whether to accept or reject, with the following requirements: 
\begin{itemize}
\item When \((G,\id, \In) \in \mathcal{L}\), then every node accepts. 
\item When \((G,\id, \In) \notin \mathcal{L}\), there is at least one vertex that rejects.
\end{itemize}

\subsection{Distributed languages for finite-state dynamics} Consider a graph \(G\), a finite-state dynamic (FSD) \(F\) over \(G\), and a configuration \(x\). The \emph{orbit} of \(x\) is the sequence of configurations \(\{x^t\}_{t>0}\) such that \(x^0 = x\) and for every \(t>0\), \(x^{t} = F(x^{t-1})\). We say that the dynamics of \(x\) \emph{converge} in at most \(k\geq0\) time-steps if \(x^k\) is a fixed point, i.e., \(x^k = F(x^k)\). We denote by \(\convergensen(k,q)\) the set of pairs \((G,F)\) satisfying that every configuration \(x\) converges in at most \(k\) time-steps. Formally,

\[\convergensen(k,q) = \left\{ ( G,F ) : \begin{array}{l} F: Q^{V(G)}\rightarrow Q^{V(G)} \textrm{ is a  FSD over }  G, \\ |Q|\leq q, \textrm{ and } \\ F(x^k) = x^k \textrm{ for every } x\in Q^{V(G)} \end{array}\right\}\]

It is easy to see that there are no local decision algorithms for \(\convergensen(k,q)\). That is, there are no algorithms in which every node of a network exchanges information solely with nodes in its vicinity and decides whether the dynamics converge within a limited number of time-steps.

\subsection{Communication Complexity} Given a Boolean function \(f: X \times Y \rightarrow \{0,1\}\), where Alice is given an input \(x \in X\) and Bob is given an input \(y \in Y\), the deterministic communication complexity of \(f\) is the minimum number of bits Alice and Bob need to exchange to compute \(f(x,y)\), over all possible deterministic communication protocols. In the non-deterministic version of communication complexity, a third party, called the prover, is allowed to send a message (called a certificate) to one or both of the communicating parties to assist in computing the function. The challenge is to determine the minimal size of such a hint that would enable the parties to compute the function with the least amount of communication between them. More formally, the non-deterministic communication complexity of a Boolean function \(f: X \times Y \rightarrow \{0,1\}\), denoted \(N^{\textrm{cc}}(f)\), is the minimum integer \(k\) such that there exists a function \(g: X \times Y \times \{0,1\}^k \rightarrow \{0,1\}\) (where \(\{0,1\}^k\) represents the hint from a prover) satisfying:

\begin{enumerate}
    \item \(f(x,y) = 1\) implies there exists a certificate \(c \in \{0,1\}^k\) such that \(g(x,y,c) = 1\),
    \item \(f(x,y) = 0\) implies that for all certificates \(c \in \{0,1\}^k\), \(g(x,y,c) = 0\).
\end{enumerate}
The nondeterministic communication complexity is then the maximum over all pairs \((x,y)\) of the sum of the number of bits exchanged and the size of the certificate. In this article, we prove our lower-bounds by reducing \(\convergensen\) to a problem \(\disj\) in communication complexity. This problem corresponds to the function \(\disj_n : 2^{[n]} \times 2^{[n]}\rightarrow \{0,1\}\) such that \[\disj_n(A,B) = 1 \textrm{ if and only if }A\cap B = \emptyset.\] The following result is shown in \cite{kushilevitz1997communication}.

\begin{proposition}\label{prop:disj}
 \(N^{\textrm{cc}}(\disj_n) = n\). 
\end{proposition}

\section{Finite-State Dynamics on Arbitrary Graphs}
In this section, we tackle problem \(\convergensen\) on arbitrary graphs. We begin by giving an upper bound on the certification size. For \(q>0\), we denote by \(|f_{\max}(n,q)|\) the maximum number of bits required to encode a local function of finite-state dynamics over an \(n\)-node graph on \(q\) states.

\begin{theorem}\label{theo:uppertrivial}
For each \(q>1\), \(\convergensen(k,q)\) admits a Proof-Labeling Scheme (PLS) with certificates of size 
\[b_{\max}(k+1) \cdot \log(\id_{\max}) \cdot |f_{\max}(n,q)|,\]
on graphs with identifiers in \([\id_{\max}]\).
\end{theorem}

\begin{proof}
Let \(G,F\) be an instance of \(\convergensen(k,q)\). The certification algorithm provides each node \(v\in V(G)\) with the following information:
\begin{itemize}
\item[\(\bullet\)] The set \(B^v_v(k+1)\) representing \(B_v(k+1)\).
\item[\(\bullet\)] The set of all local functions \(f^v_w\) of nodes \(w\) that are endpoints of edges in \(B^v_v(k+1)\). 
\end{itemize}

Observe that the certificates can be encoded in at most 
\[b_{\max}(k+1) \cdot \log(\id_{\max})\cdot |f_{\max}(n,q)|\] bits. Given the certificates, \(v\) can reconstruct all the neighborhoods and local functions of nodes at distance at most \(k+1\) from \(v\). In particular, \(v\) can determine all the nodes up to distance \(k+2\) and all the local functions of nodes up to distance \(k+1\).

In the verification algorithm, each node \(v\) checks the consistency of the information provided to its neighbors and verifies that \(v\) converges to a fixed-point in \(k\) time-steps for every configuration assigned to nodes with endpoints in \(B_v(k)\). Formally, \(v\) checks the following conditions for each node \(u\in N(v)\):
\begin{enumerate}
\item[a.] All the edges in \(B^u_u(k)\) belong to \(B^v_v(k+1)\).
\item[b.] If \(w\) is an endpoint of an edge in \(B_u(k+1)\cap B_v(k+1)\) then \(f^v_w = f^u_w\). In particular, \(f^v_v = f^u_v\).
\item[c.] For every configuration of the nodes with endpoints in \(B^v_v(k+1)\), \(v\) simulates \(k+1\) time-steps and checks that the configuration reached by \(v\) after \(k\) time-steps is a fixed point.
\end{enumerate}
Node \(v\) accepts only if all conditions are satisfied. Let us analyze now the soundness and completeness.\\

\emph{Completeness:} Let us suppose that \((G,F)\) is a \emph{yes-instance}. Clearly, if every node \(v\) receives the certificates as they are defined (i.e., \(B^v_v(k+1) = B_v(k+1)\), \(f_w^v = f_w\) for every \(w\) that is an endpoint of an edge in \(B_v(k+1)\)), then every node accepts. \\

\emph{Soundness:} Let us suppose that \((G,F)\) is a \emph{No-instance}. In other words, there exists a configuration \(x\) of \(G\) for which \(F(x^{k}) \neq x^k\). Let \(v\) be a node such that \(F(x^{k})_v \neq x^k_v\). Assuming that every node accepts conditions a. and b., and since these conditions are satisfied for every node, we have that \(B^v_v(k+1) = B_v(k+1)\) and \(f_w^v = f_w\) for every \(w\) that is an endpoint of an edge in \(B_v(k+1)\). Observe that \(x^{k+1}_v\) only depends on the initial configuration of nodes at distance at most \(k+2\) from \(v\), and the local functions of nodes at distance at most \(k+1\) from \(v\). Therefore, \(v\) rejects the condition c of the verification algorithm.
\qed
\end{proof}

Notice that when \(b_{\max}(k+1) = \mathcal{O}(n^2)\) (for instance, on dense graphs of small diameter), the upper bound above is \(\mathcal{O}(n^2\log(\id_{\max})|f_{\max}|)\), which is greater than the trivial upper bound \(\mathcal{O}(n^2\log(\id_{\max}) + n|f_{\max}|)\) that involves providing each node with all the edges of the graph and all the local functions. In the next theorem, we show that, up to logarithmic factors, there is also a quadratic lower bound for the convergence in at most two time-steps, and specifically for dynamics on four states that admit a succinct representation.

\begin{theorem}\label{theo:upperboundunbounded}
Every Proof Labeling Scheme for \(\convergensen(2,4)\) has certificates of size \(\Omega(n^2/(\log n))\). This holds even when the problem 
is restricted to local functions that admit a succinct representation. 
\end{theorem}

\begin{proof}
We reduce \(\convergensen(2,4)\) to \disj\ in two-player communication complexity. Given an instance of \disj, we build an instance of problem \(\convergensen(2,4)\), which is composed of a \emph{lower-bound graph} and a dynamic picked from a family of \emph{lower-bound dynamics}.

We begin by giving a high-level description of the construction. Let \(n\) be a positive integer, and let \({\bf A},{\bf B} \subseteq \binom{[n]}{2}\) be a pair of sets, interpreted as an instance of \disj. The set \(\bf A\) is associated with the graph \(G_{\bf A}\) with the vertex set \(\{v_1,\dots, v_n\}\) such that, for each \(i,j\in [n]\), node \(v_i\) is adjacent to \(v_j\) if and only if \(\{i,j\} \in {\bf A}\). The graph \(G_{\bf B}\) is defined analogously using the set \({\bf B}\) instead of \({\bf A}\).

The lower-bound graph has vertices containing disjoint copies of \(G_{\bf A}\) and \(G_{\bf B}\), as well as a \emph{bit gadget}, which is connected to every other vertex. A configuration satisfies an \emph{admissibility condition} when the state of the bit gadget encodes the binary representation of an element of \([n]\times [n]\), and this element represents an edge that is both in \(G_{\bf A}\) and \(G_{\bf B}\) (i.e., when \({\bf A}\cap{\bf B} \neq \emptyset\)). The dynamic of the system is designed to oscillate in a limit-cycle of period two only when the admissibility condition is satisfied. When the admissibility condition is not satisfied, the dynamic reaches a fixed point in at most two time-steps.\\

\noindent {\it Lower-bound graph.} For each \(n>0\), the \emph{lower-bound graph} \(G({\bf A,B})\) is a graph of size \(2n + 4\lceil \log n \rceil\) (see Figure~\ref{fig:lowerboundgraphTh1}). The vertex set of \(G({\bf A,B})\) is partitioned into four subsets, namely, \(V_A = \{v^A_1, \dots, v^A_n\}\), \(V_B = \{v^B_1, \dots, v^B_n\}\), \(D_A = \{d^A_1, \dots, d^A_\ell\}\) and \(D_B = \{d^B_1, \dots, d^B_\ell\}\), where \(\ell = 2\lceil \log n \rceil\). The edge set of \(G_{\bf A,B} = (V,E)\) contains all the edges with one endpoint in \(V_A\) and the other in \(D_A\), all edges between nodes in \(V_B\) and \(D_B\), and all edges with endpoints in \(D_A\cup D_B\). It also contains, for each \(i,j \in [n]\) the edge \(\{v_i^A, v_j^A\}\) (respectively, \(\{v_i^B, v_j^B\}\)) if node \(v_i\) is adjacent to \(v_j\) in \(G_{\bf A}\) (respectively, \(G_{\bf B}\)).

Nodes in \(D_A\cup D_B\) are called the \emph{bit gadget} and the nodes of \(V_A\cup V_B\) are the \emph{set gadget}.

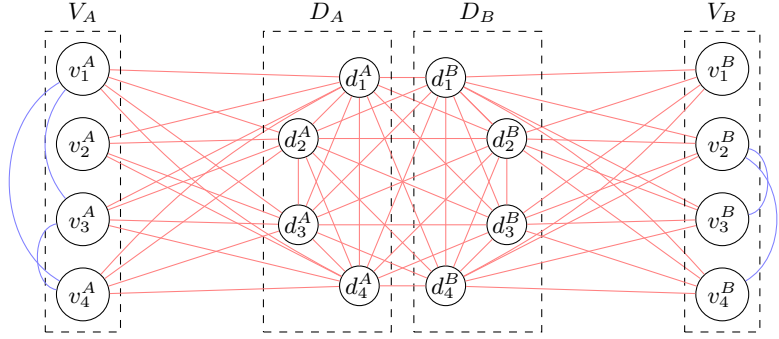
\begin{figure}
\centering
\begin{tikzpicture}
    \tikzstyle{vertex}=[draw,fill=white!15,circle,minimum size=20pt,inner sep=0pt]
    \tikzstyle{vertex2}=[draw,fill=white!15,circle,minimum size=15pt,inner sep=0pt]
    \tikzstyle{red edge} = [draw,red!50]
    \tikzstyle{blue edge} = [draw,blue!50]
\node [vertex] (v5) at (0,7.5) {$v_{1}^{A}$};
\node [vertex] (v6) at (0,6.5) {$v_{2}^{A}$};
\node [vertex] (v7) at (0,5.5) {$v_{3}^{A}$};
\node [vertex] (v8) at (0,4.5) {$v_{4}^{A}$};
\begin{scope}[ shift={(4.25,6)}]
\node[vertex2] (v3) at (67.5:1.5) {$d_{1}^{B}$};
\node[vertex2] (v4) at (22.5:1.5) {$d_{2}^{B}$};
\node[vertex2] (v15) at (-22.5:1.5) {$d_{3}^{B}$};
\node[vertex2] (v16) at (-67.5:1.5) {$d_{4}^{B}$};.
\node[vertex2] (v1) at (112.5:1.5) {$d_{1}^{A}$};
\node[vertex2] (v2) at (157.5:1.5) {$d_{2}^{A}$};
\node[vertex2] (v13) at (-157.5:1.5) {$d_{3}^{A}$};
\node[vertex2] (v14) at (-112.5:1.5) {$d_{4}^{A}$};
\node at (0,0) {};
\end{scope}
\node [vertex] (v9) at (8.5,7.5) {$v_{1}^{B}$};
\node [vertex] (v10) at (8.5,6.5) {$v_{2}^{B}$};
\node [vertex] (v11) at (8.5,5.5) {$v_{3}^{B}$};
\node [vertex] (v12) at (8.5,4.5) {$v_{4}^{B}$};
\draw [red edge] (v1) edge (v2);
\draw [red edge] (v3) edge (v4);
\draw [red edge] (v4) edge (v1);
\draw [red edge] (v1) edge (v3);
\draw [red edge] (v3) edge (v2);
\draw [red edge] (v4) edge (v2);
\draw [red edge] (v5) edge (v2);
\draw [red edge] (v5) edge (v1);
\draw [red edge] (v6) edge (v1);
\draw [red edge] (v6) edge (v2);
\draw [red edge] (v2) edge (v7);
\draw [red edge] (v7) edge (v1);
\draw [red edge] (v2) edge (v8);
\draw [red edge] (v1) edge[bend right=10] (v8);
\draw [red edge] (v3) edge (v9);
\draw [red edge] (v9) edge (v4);
\draw [red edge] (v10) edge (v3);
\draw [red edge] (v10) edge (v4);
\draw [red edge] (v11) edge (v3);
\draw [red edge] (v11) edge (v4);
\draw [red edge] (v3) edge[bend left=10] (v12);
\draw [red edge] (v4) edge (v12);
\draw [dashed] (-0.5,8) rectangle (0.5,4);
\draw [dashed] (2.4,8) rectangle (4.1,4);
\draw [dashed] (4.4,8) rectangle (6.1,4);
\draw [dashed] (8,8) rectangle (9,4);
\node at (0,8.25) {$V_A$};
\node at (3.25,8.25) {$D_A$};
\node at (5.25,8.25) {$D_B$};
\node at (8.5,8.25) {$V_B$};
\draw [blue edge] (v5) edge[bend right=40] (v7);
\draw [blue edge] (v5) edge[bend right=60] (v8);
\draw [blue edge] (v7) edge[bend right=80] (v8);
\draw [blue edge] (v11) edge[bend right=80]  (v10);
\draw [blue edge] (v10) edge[bend left=60]  (v12);
\node at (4.25,6) {};
\draw [red edge] (v13) edge (v5);
\draw [red edge] (v13) edge (v6);
\draw [red edge] (v13) edge (v7);
\draw [red edge] (v13) edge (v8);
\draw [red edge] (v14) edge[bend left=10] (v5);
\draw [red edge] (v14) edge (v7);
\draw [red edge] (v8) edge (v14);
\draw [red edge] (v14) edge (v6);
\draw [red edge] (v9) edge (v15);
\draw [red edge] (v15) edge (v10);
\draw [red edge] (v11) edge (v15);
\draw [red edge] (v15) edge (v12);
\draw [red edge] (v16) edge[bend right=10] (v9);
\draw [red edge] (v16) edge (v10);
\draw [red edge] (v11) edge (v16);
\draw [red edge] (v16) edge (v12);
\draw [red edge] (v13) edge (v1);
\draw [red edge] (v1) edge (v14);
\draw [red edge] (v1) edge (v16);
\draw [red edge] (v1) edge (v15);
\draw [red edge] (v2) edge (v13);
\draw [red edge] (v2) edge (v14);
\draw [red edge] (v16) edge (v2);
\draw [red edge] (v15) edge (v2);
\draw [red edge] (v14) edge (v13);
\draw [red edge] (v13) edge (v16);
\draw [red edge] (v13) edge (v15);
\draw [red edge] (v13) edge (v4);
\draw [red edge] (v13) edge (v3);
\draw [red edge] (v14) edge (v16);
\draw [red edge] (v14) edge (v15);
\draw [red edge] (v14) edge (v4);
\draw [red edge] (v14) edge (v3);
\draw [red edge] (v15) edge (v16);
\draw [red edge] (v16) edge (v4);
\draw [red edge] (v16) edge (v3);
\draw [red edge] (v15) edge (v4);
\draw [red edge] (v15) edge (v3);
\draw [red edge] (v4) edge (v3);
\end{tikzpicture}
\caption{Lower-bound graph \(G({\bf A,B})\) when \(n = 4\), \({\bf A} = \{\{1,3\}, \{1,4\}, \{3,4\}\}\) and \({\bf B} = \{\{2,3\}, \{2,4\}\}\) }
\label{fig:lowerboundgraphTh1}
\end{figure}

\medskip

\noindent{\it States of the system.} 
  For all nodes the set of states is 
\( Q = \{0,1\} \times \{0,1\}\).
Given a node \(u\) in state \(s(u) = (m(u),c(u)) \in Q\) we say that \(m(u)\) is the \emph{mark}  of \(u\) and \(c(u)\) is the \emph{clock} of \(u\). \\

Given the set of states \(Q\), we define now lower-bound dynamics given by lower-bound local functions, that we call \(F(\bf{A,B})\). The local functions of nodes in \(D_A \cup D_B\) will be independent on the input of the players, while the ones nodes of \(V_A\) (respectively \(V_B\))  depend only on \({\bf A}\) (respectively \({\bf B}\)).  We describe first the local functions of nodes in \(D_A\cup D_B\).\\

\noindent{\it Local functions of nodes of the bit gadget.} First, we define the following \emph{admissibility conditions} for a given configuration, which is checked for every node in \(D_A\) (respectively \(D_B\)): 
\begin{itemize}
\item[d1)] Every node in \(D_A\cup D_B\) has the same clock.
\item[d2)] There are exactly two nodes in \(V_A\) (resp. \(V_B\)) with mark \(1\). 
\item[d3)] For each \(s \in \{1, \dots, \ell\}\), the  mark of \(d^A_s\) equals the one of \(d^B_s\). 
\end{itemize}

The dynamic of the nodes \(u\) in \(D_A\) and \(D_B\) is then defined as follows: if the admissibility conditions are satisfied, then the mark of \(u\) remains unchanged, and the clock of \(u\) switches to \(0\) if the clock was \(1\) and vice-versa. If the admissibility condition is not satisfied, the whole state of \(u\) (mark and clock) remains unchanged. In any case, the marks in the nodes of the bit gadget do not change under any circumstances.

For each \(s\in [\ell]\), the local function of \(d_s^A\) (respectively \(d_s^B\)) is succinct. Indeed, to define the function, we simply need to indicate which neighbors of \(d_s^A\) (resp. \(d_s^B\)) belong to \(D_A\) and which to \(V_A\) (resp. \(D_B\) and \(V_B\)), and which node is \(d_s^B\) (resp. \(d_s^A\)). This can be encoded using \(\cO(1)\) bit per neighbor.

Observe that while the admissibility condition is verified, the clocks in the nodes of \(D_A \cup D_B\) continuously switch between \(0\) and \(1\). Suppose that in the configuration on a given time-step \(t\), there are nodes \(d_1, d_2 \in D_A\cup D_B\) such that the admissibility condition is satisfied for \(d_1\) but not for \(d_2\). Then, in time-step \(t\), we have that \(c(d_1) = c(d_2)\), but in \(t+1\), \(c(d_1) \neq c(d_2)\). This implies that in time-step \(t+1\), the admissibility condition (d1) is not satisfied by any node in \(D_A \cup D_B\), which results in every node in that set being fixed in its state forever.\\

\noindent {\it Local functions of nodes of the set gadgets.} Let \(i\in [n]\). The following admissibility conditions are considered for node~\(v_i^A\):

\begin{itemize}
\item[v1)] The mark of \(v_i^A\) is \(1\).
\item[v2)] Exactly one neighbor of \(v_i^A\) in \(V_A\) has mark \(1\).
\item[v3)] The mark of \(d^A_1, \dots, d^A_{\ell/2}\) or \(d^A_{\ell/2+1}, \dots, d^A_{\ell}\) represents the binary representation of \(i\).
\end{itemize}

The local function of \(v_i^A\) indicates that the mark of node \(v^A_i\) is \(1\) if the admissibility conditions are satisfied, and \(0\) otherwise. The clock of \(v_i^A\) does not change under any circumstance.

The local function of \(v_i^A\) is succinct for every \(i\in [n]\). To define the function, it is necessary to simply indicate which neighbors of \(v_i^A\) belong to \(D_A\) and which to \(V_A\). This can be encoded using \(\cO(1)\) bit per neighbor. Additionally, specifying the index \(i\) requires \(\cO(\log n)\) bits. Analogous admissibility conditions and local functions are defined for node~\(v_i^B\) by switching the subindices \(A\) by \(B\).

Independently of the initial configuration, the admissibility condition is satisfied for at most two nodes in \(V_A\), as the marks in \(D_A\) can represent at most two indices in \([n]\). Then, in the first time-step, the mark of every node in \(V_A\) is \(0\), except for at most two nodes. The mark of the remaining nodes, namely \(v^A_i\) and \(v_j^A\), is \(1\) only if the admissibility condition is satisfied for two nodes. In particular, condition (v2) implies that \(v_i^A\) is adjacent to \(v_j^A\). Hence \(\{i,j\}\in {\bf A}\). In any case, the state of every node of \(V_A \cup V_B\) remains unchanged after the first time-step.

\noindent{\it The reduction.} We now show that \((G({\bf A,B}), F({\bf A,B}))\) converges in at most two time-steps if and only if \({\bf A}\cap {\bf B} =  \emptyset\).  Let us suppose that \({\bf A}\cap {\bf B} \neq \emptyset\), and let \(\{i,j\} \in {\bf A} \cap {\bf B}\). We define the following initial configuration \(x\).

\begin{itemize}
\item[$\bullet$] For every \(s \in [n]\setminus \{i,j\}\), nodes \(v_s^A\) and \(v_s^B\) have initial configuration \((0,0)\). 
\item[$\bullet$] The initial configuration of \(v_i^A\), \(v_j^A\), \(v_i^B\), and \(v_j^B\) is \((1,0)\).
\item[$\bullet$] For each \(s \in [\ell/2]\), the initial configuration of \(d_s^A\) and \(d_s^B\) is \((b,0)\), where \(b\) is the \(s\)-th bit in the binary representation of \(i\). 
\item[$\bullet$] For each \(s \in [\ell/2+1,\ell]\), the initial configuration of \(d_s^A\) and \(d_s^B\) is \((b,0)\), where \(b\) is the \(s\)-th bit in the binary representation of \(j\).
\end{itemize}

Observe that in \(x\) all the admissibility conditions are satisfied. Moreover, in the next time-steps, the mark of every node in \(G({\bf A,B})\) remains unchanged, and the clocks of every node in \(D_A \cup D_B\) switch between \(0\) and \(1\) back and forth. Therefore, the dynamic \((G({\bf A,B}), F({\bf A,B}))\) does not converge.
\medskip

Now let us suppose that \({\bf A}\cap {\bf B}= \emptyset\). Let \(x\) be any initial configuration, and let \(y\) be the configuration obtained in the first time-step. Observe that the state of every node in \(V_A \cup V_B\) is fixed in the state of \(y\) in the next time-steps. If fewer than two nodes in \(V_A\) (respectively \(V_B\)) have mark \(1\) in \(y\), all nodes in \(D_A\) (respectively \(D_B\)) are fixed in their state forever. Then, in the third time-step, all nodes in \(D_B\) (respectively \(D_A\)) are also fixed. Suppose then that in \(y\), exactly two nodes \(v_i^A, v_j^A \in V_A\) and two nodes \(v_p^B, v_q^B \in V_B\) have mark \(1\).  By condition (v3), in \(y\) the marks of the nodes in \(D_A\) are the binary representation of \(i\) and \(j\). By condition (v2), \(v_i^A\) and \(v_j^A\) are adjacent. Hence \(\{i,j\}\in {\bf A}\). Similarly, the marks on \(D_B\) are the binary representations of \(p\) and \(q\), and \(v_p^B\) is adjacent to \(v_q^B\). Hence \(\{p,q\}\in {\bf B}\). Since \({\bf A}\cap {\bf B}= \emptyset\), condition (d3) is not satisfied in \(y\), implying that \(y\) is a fixed point. 

We deduce that \((G, F({\bf A,B})) \in \convergensen(2,4)\) if and only if \({\bf A}\cap {\bf B}= \emptyset\). Let \(\pi\) be a PLS for \(\convergensen(2,4)\). We define the following two-player protocol \(\mathcal{P}\) for \disj. On instance \({\bf A}\), Alice computes nondeterministically the certificates that \(\pi\) would give on all nodes in \(V_A \cup D_A \cup D_B\), and simulates the verification protocol of \(\pi\) on all nodes in \(V_A\cup D_A\). Alice communicates a single bit to Bob indicating if every node in \(V_A\cup D_A\) has accepted, as well as all the certificates of \(D_A \cup D_B\). Analogously, Bob computes the certificates of \(D_A\cup D_B \cup V_B\) and simulates the  verification protocol of \(\pi\) on all nodes in \(V_B \cup D_B\). Bob accepts if all nodes in \(V_B\) accept and his certificates for \(D_A \cup D_B\) are the same as those generated by Alice. The correctness of \(\mathcal{P}\) follows directly from the soundness and completeness of \(\pi\). Let \(C(n)\) be the maximum size of a certificate produced by \(\pi\) on graphs of size \(n\). According to Proposition~\ref{prop:disj}, it follows that \(C(2n + 2\lceil \log n \rceil) \cdot \lceil \log n \rceil = \Omega(n^2)\), implying that \(C(n) = \Omega(n^2/\log n)\).
\qed
\end{proof}

\section{Finite-State Dynamics in Graphs of Bounded Degree}

Given that the problem is hard in general graphs, we focus our study on finite-state dynamics over graphs of bounded degree. For such graphs, all local functions are succinct. Moreover, an analysis of the bound given by Theorem~\ref{theo:uppertrivial} provides a non-trivial upper-bound for the certificate size.

\begin{corollary}\label{coro:max3}
For each \(q>1\) and for each \(\Delta > 2\), problem \(\convergensen(k,q)\) admits a Proof-Labeling Scheme with certificates of size at most
\[\Delta((\Delta -1)^{k+1}-1)\cdot q^{\Delta+1} \log (q) \cdot \log (\id_{\max})\]
on graphs of maximum degree \(\Delta\) with assignments of identifiers in \([\id_{\max}]\).
\end{corollary}

\begin{proof}
Observe that in graphs of maximum degree \(\Delta\), there are at most \((\Delta -1)^{k+1} -1\) nodes at a distance of at most \(k+1\) from \(v\). Then, we have that \(b_{\max}(k+1)\) is \(\Delta((\Delta -1)^{k+1} -1)\). Also, when restricted to graphs of maximum degree \(\Delta\), a local function on \(q\) states of a node can be encoded in \(q^{\Delta+1} \log q\) bits by writing up the table of values of the function when the inputs are ordered increasingly by identifier. We deduce the result by combining these bounds with Theorem~\ref{theo:uppertrivial}.
\qed
\end{proof}

In the following result, we show that the exponential dependency on \(k\) is necessary even on graphs of bounded degree. To prove our result, we need to introduce a \emph{binary decoder}, which is a specific type of Boolean circuit. A \emph{Boolean circuit} is a directed acyclic graph \(C\) where each node (called \emph{gate}) is assigned a Boolean function. The nodes of in-degree zero are called \emph{input gates} and the nodes with out-degree zero are called \emph{output gates}. A truth assignment of the input gates induces a truth assignment of every other gate in the circuit. Every node takes the value of its assigned Boolean function, using the truth values of its incoming neighbors as arguments. The number of gates in a Boolean circuit is known as its \emph{size}, while the maximum distance between an input and an output gate is called its \emph{depth}. A Boolean circuit \(C\) is called a \emph{binary decoder} if there exists a positive integer \(t\) such that \(C\) has \(t\) inputs and \(2^t\) outputs, named \(\{v_1, \dots, v_{2^t}\}\). This Boolean circuit satisfies that for each \(i\in [2^t]\), the output value of \(v_{i}\) is {\sf True} if and only if the truth values of the inputs (mapped to \(0/1\) values) correspond to the binary representation of \(i\).

\begin{lemma}~\label{lem:binarydecoder}
For each positive integer \(t\), there exists a binary decoder \(C_t\) satisfying the following conditions: 
\begin{itemize}
\item[(1)] each gate of \(C_t\) has a total degree (in-degree plus out-degree) at most \(3\); 
\item[(2)] each input gate of \(C_t\) has out-degree at most \(2\); 
\item[(3)] each output gate of \(C_t\) has in-degree at most \(2\); 
\item[(4)] the depth of \(C_t\) is \(3t\); and 
\item[(5)] the number of gates of \(C_t\) is \((8t-3)2^{t-1}-t\). 
\end{itemize}
\end{lemma}

\begin{proof}
In our construction, we use two types of gadgets, namely a \emph{duplicator}, and an \emph{expander}. A \(2\)-duplicator is a circuit with one input gate and two output gates that simply replicate the value of the input gate. This can be implemented with depth \(1\) using one-input disjunctions in the output gates. We can also implement a \(2^t\)-duplicator, which consists of a Boolean circuit with one input gate and \(2^t\) output gates, all replicating the value of the input gate. The \(2^t\)-duplicator can be implemented by arranging \(2\)-duplicators in a binary tree of depth \(t\). A \(2^t\)-duplicator satisfies conditions (1)-(3) and has \(2^{t+1}-1\) gates.

The expander is similar to a \(2\)-duplicator, except that the first output gate is a negation. More precisely, the expander has one input gate and two output gates. The first output gate negates the value of the input, while the second replicates the truth value of the input gate.

Our construction is done by induction. For each positive integer \(t\), we denote by \(\inp_t(d_1), \dots, \inp_t(d_t)\) the input gates and \(\out_t(v_1), \dots, \out_t(v_{2^t})\) the output gates of \(C_t\).
The circuit \(C_1\) simply consists of a gate \(\inp_1(d_1)\) connected to an expander where one output is assigned to \(\out_1(v_1)\) and the other to \(\out_1(v_2)\). Clearly, \(C_1\) is a binary decoder satisfying (1)-(5).

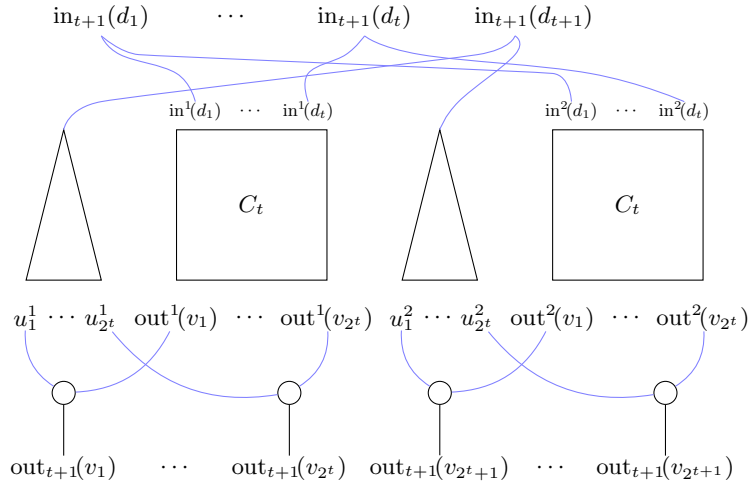
\begin{figure}[ht]
\centering
\begin{tikzpicture}
    \tikzstyle{vertex}=[draw,fill=white!15,circle,minimum size=20pt,inner sep=0pt]
    \tikzstyle{red edge} = [draw,red!50]
    \tikzstyle{blue edge} = [draw,blue!50]

\draw (5,4.5) -- (5.5,6.5) node[inner sep=0] (v8) {} -- (6,4.5)--cycle;
\draw (0,4.5) -- (0.5,6.5) node[inner sep=0] (v9) {} -- (1,4.5)--cycle;
\draw (2,4.5) rectangle (4,6.5);
\draw (7,6.5) rectangle (9,4.5);
\node (v1) at (1,8) {$\textrm{in}_{t+1}(d_1)$};
\node at (2.75,8) {$\cdots$};
\node (v3) at (4.5,8) {$\textrm{in}_{t+1}(d_t)$};
\node (v7) at (6.75,8) {$\textrm{in}_{t+1}(d_{t+1})$};
\node[scale=0.75] (v2) at (2.25,6.75) {$\textrm{in}^1\!(d_1)$};
\node[scale=0.75] at (3,6.75) {$\cdots$};
\node[scale=0.75] (v4) at (3.75,6.75) {$\textrm{in}^1\!(d_t)$};
\node[scale=0.75] (v5) at (7.25,6.75) {$\textrm{in}^2\!(d_1)$};
\node[scale=0.75] at (8,6.75) {$\cdots$};
\node[scale=0.75] (v6) at (8.75,6.75) {$\textrm{in}^2\!(d_t)$};
\node[inner sep=0] (v10) at (0,4) {$u^1_{1}$};
\node[inner sep=0] at (0.5,4) {$\cdots$};
\node[inner sep=0] (v13) at (1,4) {$u^1_{2^t}$};

\node[inner sep=0] (v12) at (2,4) {$\textrm{out}^1\!(v_1)$};
\node[inner sep=0] at (3,4) {$\cdots$};
\node[inner sep=0] (v15) at (4,4) {$\textrm{out}^1\!(v_{2^t})$};

\node[inner sep=0] (v16) at (5,4) {$u^2_{1}$};
\node[inner sep=0] at (5.5,4) {$\cdots$};
\node[inner sep=0] (v19) at (6,4) {$u^2_{2^t}$};

\node[inner sep=0] (v18) at (7,4) {$\textrm{out}^2\!(v_1)$};
\node[inner sep=0] at (8,4) {$\cdots$};
\node[inner sep=0] (v21) at (9,4) {$\textrm{out}^2\!(v_{2^t})$};

\node[inner sep=0] (v22) at (0.5,2) {$\textrm{out}_{t+1}\!(v_1)$};
\node[inner sep=0] at (2,2) {$\cdots$};
\node[inner sep=0] (v23) at (3.5,2) {$\textrm{out}_{t+1}\!(v_{2^t})$};
\node[inner sep=0] (v24) at (5.5,2) {$\textrm{out}_{t+1}\!(v_{2^t+1})$};
\node[inner sep=0] (v25) at (8.5,2) {$\textrm{out}_{t+1}\!(v_{2^{t+1}})$};

\node at (7,2) {$\cdots$};

\node[circle,draw] (v11) at (0.5,3) {};
\node[circle,draw] (v14) at (3.5,3) {};
\node[circle,draw] (v17) at (5.5,3) {};
\node[circle,draw] (v20) at (8.5,3) {};
\draw [blue edge] (v10) edge[bend right] (v11);
\draw [blue edge] (v12) edge[bend left] (v11);
\draw [blue edge] (v13) edge[bend right] (v14);
\draw [blue edge] (v15) edge[bend left] (v14);
\draw [blue edge] (v16) edge[bend right] (v17);
\draw [blue edge] (v18) edge[bend left] (v17);
\draw [blue edge] (v19) edge[bend right] (v20);
\draw [blue edge] (v21) edge[bend left] (v20);
\draw  (v11) edge (v22);
\draw  (v14) edge (v23);
\draw  (v17) edge (v24);
\draw  (v20) edge (v25);
\draw [blue edge] plot[smooth, tension=.3] coordinates {(6.5,7.75)  (6,7.5) (1,6.875) (v9)};
\draw [blue edge] plot[smooth, tension=.7] coordinates {(1,7.75) (1.25,7.5) (2,7.25) (2.25,6.875)};
\draw [blue edge] plot[smooth, tension=.15] coordinates {(1,7.75) (1.5,7.5) (7,7.25) (7.25,6.875)};
\draw [blue edge] plot[smooth, tension=.7] coordinates {(4.5,7.75) (4.25,7.5) (3.75,7.25) (3.75,6.875)};
\draw [blue edge](4.5,7.75);
\draw [blue edge] plot[smooth, tension=.7] coordinates {(4.5,7.75) (5.25,7.5) (7.5,7.25) (8.75,6.875)};
\draw [blue edge] plot[smooth, tension=.7] coordinates {(6.5,7.75)};
\draw [blue edge] plot[smooth, tension=.7] coordinates {(6.5,7.75) (6.5,7.5) (5.75,6.875) (v8)};
\node at (3,5.5) {$C_t$};
\node at (8,5.5) {$C_t$};
\end{tikzpicture}
\caption{Schema of the construction of the binary decoder \(C_{t+1}\).}
\label{fig:binarydecoder}
\end{figure}

Suppose we have a binary decoder construction \(C_t\) satisfying conditions (1)-(5) for a given value \(t\). We now describe the construction of \(C_{t+1}\) (see also Figure~\ref{fig:binarydecoder}). First, for each \(i \in [t]\), we create a copy of a \(2\)-duplicator and identify the input of each duplicator with the input gate \(\inp_{t+1}(d_i)\). Let \(d_i^1\) and \(d_i^2\) be the output gates of the \(i\)-th duplicator. Then, we create two copies of \(C_t\). For \(j\in \{1,2\}\), we denote \(\inp^j_t(d_1), \dots, \inp^j_t(d_t)\) and \(\out^j_t(v_1), \dots, \out^j_t(v_{2^t})\) as the input and output gates of the \(j\)-th copy of \(C_t\). We identify gate \(d_i^j\) with \(\inp^j_t(d_i)\).

Then, we create a copy of an expander, and identify its input with \(\inp_{t+1}(v_{t+1})\). We name \(d^1_{t+1}\) and \(d^2_{t+1}\) as the two outputs of the expander. We attach to each \(d^1_{t+1}\) a \(2^t\)-duplicator. For each \(j\in \{1,2\}\), we denote \(u^j_{1}, \dots, u^j_{2^t}\) as the outputs of the \(2^t\)-duplicator with input \(d^j_{t+1}\). The values of \(u^1_{1}, \dots, u^1_{2^t}\) are the negation of \(\inp_{t+1}(d_{t+1})\), while the values of \(u^2_{1}, \dots, u^2_{2^t}\) are equal to the value of \(\inp_{t+1}(d_{t+1})\).

All output gates of \(C_{t+1}\) are conjunctions. For each \(i \in [2^t]\), the output gate \(\out_{t+1}(v_i)\) receives inputs from gates \(\out_{t}^1(v_i)\) and \(u^1_i\), and the output gate \(\out_{t+1}(v_{i+2^t})\) receives inputs from gates \(\out_{t}^2(v_i)\) and \(u^2_i\). Finally, to maintain condition (1) in the induction step, we ensure that the output gates have in-degree \(1\). This is achieved by duplicating all output gates and connecting each output gate with its duplicate.

We now show that \(C_{t+1}\) is a binary decoder. Consider a truth-assignment of the input values, and let \(a \in [2^{t+1}]\) be the integer represented in binary by the inputs of \(C_{t+1}\). Let \(b \in [2^t]\) be the integer that their binary truth values represent. Observe that \(a=b\) when \(\inp_{t+1}(d_{t+1})\) is \(\sf False\) and \(a = b + 2^t\) when \(\inp_{t+1}(d_{t+1})\) is \(\sf True\). According to our construction and the induction hypothesis, we have that the truth value of \(\out^1_{t}(v_i)\) and \(\out^2_{t}(v_i)\) is \(\sf{False}\) for each \(i \neq b\), and \(\sf{True}\) if \(i=b\). We obtain that for each \(i\in [2^t] \setminus \{b\}\), the values of \(\out_{t+1}(v_{i})\) and \(\out_{t+1}(v_{i+2^t})\) are \(\sf False\), as each one of these gates receives an input from node \(\out_{t}(v^j_i)\) which are all \(\sf False\).

 Assume that \(\inp_{t+1}(d_{t+1})\) is \(\sf False\) (hence \(a=b\)). Then \(u^1_1,\dots, u^1_{2^t}\) are \(\sf True\) and all \(u^2_1,\dots, u^2_{2^t}\) are \(\sf False\). The truth value of \(\out_{t+1}(v_{2^t + b})\) is {\sf False} because \(u^2_{b}\) is \(\sf False\), and the truth value of \(\out_{t+1}(v_{b})\) is {\sf True} as \(u^1_{b}\) and \(\out_t(v_b^1)\) are \(\sf True\). Similarly, when \(\inp_{t+1}(d_{t+1})\) is \(\sf True\) (hence \(a = b + 2^t\)), \(u^1_1,\dots, u^1_{2^t}\) are \(\sf False\) and all \(u^2_1,\dots, u^2_{2^t}\) are \(\sf True\). Then, the truth value of \(\out_{t+1}(v_{b})\) is {\sf False} as \(u^1_{b}\) is \(\sf False\), and the truth value of \(\out_{t+1}(v_{2^t + b})\) is {\sf True} as \(u^2_{b}\) and \(\out_t(v_b^1)\) are \(\sf True\). Only the output \(\out_{t+1}(v_a)\) is {\sf True}, hence \(C_{t+1}\) is a binary decoder.

By the induction hypothesis and the definition of duplicators and expanders, our construction of \(C_{t+1}\) satisfies conditions (1)-(3). The depth of \(C_{t+1}\) is the depth of \(C_{t}\) plus \(3\), satisfying condition (4). Finally, the number of gates of \(C_{t+1}\) is twice the number of gates of \(C_{t}\) plus \(t+1\) input gates, two \(2^t\)-duplicators, \(2^{t+1}\) output gates and their duplicates, resulting in \[2((8t-3)2^{t-1}-t) + (t+1) + 2(2^{t+1}-1) + 2\cdot 2^{t+1} = (8(t+1)-3)2^{(t+1)-1} - (t+1) \] gates. This implies that \(C_{t+1}\) satisfies condition (5), completing the proof. \qed 
\end{proof}

\begin{theorem}\label{theo:lowerbounddegree3}
Let \(k \geq 2\). Every PLS for \(\convergensen(k,3)\) requires certificates of size at least \(2^{k/6} \cdot 6/{k}\), even when the problem is restricted to input graphs of degree at most \(3\). 
\end{theorem}

\begin{proof}
Our proof follows similar ideas to the proof of Theorem~\ref{theo:upperboundunbounded}. 
We reduce problem \(\convergensen(k,3)\) to \disj\ in two-player communication complexity. Our construction considers three different gadgets: the \emph{selector}, \emph{collector}, and \emph{binary decoder}. In this proof, we explain the lower bound graph (see Figure~\ref{fig:lowerboundgraphTh3}) alongside the corresponding local functions.

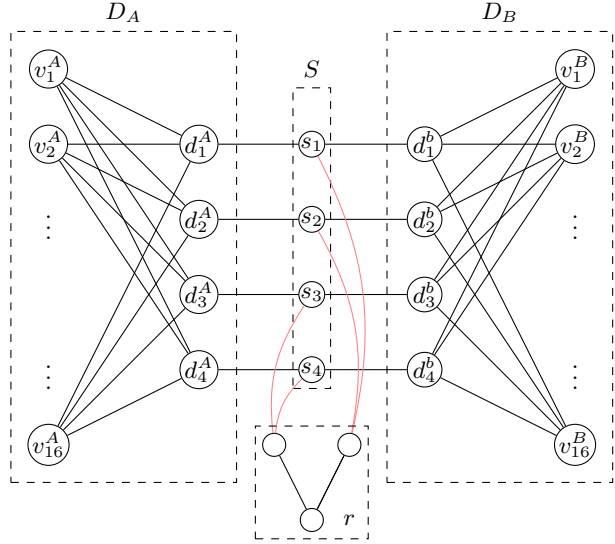
\begin{figure}[ht]
\centering
\begin{tikzpicture}
    \tikzstyle{vertex}=[draw,fill=white!15,circle,minimum size=20pt,inner sep=0pt]
    \tikzstyle{red edge} = [draw,red!50]
    \tikzstyle{blue edge} = [draw,blue!50]

\node[draw,inner sep=0,circle] (v1) at (0,8.5) {$v^A_{1}$};
\node[draw,inner sep=0,circle] (v7) at (0,7.5) {$v^A_{2}$};
\node[inner sep=0,circle] at (0,6.5) {$\vdots$};
\node[inner sep=0,circle] at (0,5.5) {};
\node[inner sep=0,circle] at (0,4.5) {$\vdots$};
\node[draw,inner sep=0,circle] (v3) at (0,3.5) {$v^A_{16}$};
\node[draw,inner sep=0,circle] (v2) at (2,7.5) {$d^A_{1}$};
\node[draw,inner sep=0,circle] (v5) at (2,6.5) {$d^A_{2}$};
\node[draw,inner sep=0,circle] (v6) at (2,5.5) {$d^A_{3}$};
\node[draw,inner sep=0,circle] (v4) at (2,4.5) {$d^A_{4}$};
\node[draw,inner sep=0,circle] (v15) at (3.5,7.5) {$s_1$};
\node[draw,inner sep=0,circle] (v16) at (3.5,6.5) {$s_2$};
\node[draw,inner sep=0,circle] (v17) at (3.5,5.5) {$s_3$};
\node[draw,inner sep=0,circle] (v18) at (3.5,4.5) {$s_4$};
\node[draw,inner sep=0,circle] (v8) at (5,7.5) {$d^b_{1}$};
\node[draw,inner sep=0,circle] (v11) at (5,6.5) {$d^b_{2}$};
\node[draw,inner sep=0,circle] (v12) at (5,5.5) {$d^b_{3}$};
\node[draw,inner sep=0,circle] (v13) at (5,4.5) {$d^b_{4}$};
\node[draw,inner sep=0,circle] (v9) at (7,8.5) {$v^B_{1}$};
\node[draw,inner sep=0,circle] (v10) at (7,7.5) {$v^B_{2}$};
\node[inner sep=0,circle] at (7,6.5) {$\vdots$};
\node[inner sep=0,circle] at (7,5.5) {};
\node[inner sep=0,circle] at (7,4.5) {$\vdots$};
\node[draw,inner sep=0,circle] (v14) at (7,3.5) {$v^B_{16}$};
\draw [dashed] (-0.5,9) rectangle (2.5,3);
\draw [dashed] (4.5,9) rectangle (7.5,3);
\draw [dashed] (3.25,8.25) rectangle (3.75,4.25);
\draw  (v1) edge (v2);
\draw  (v3) edge (v4);
\draw  (v1) edge (v5);
\draw  (v1) edge (v6);
\draw  (v1) edge (v4);
\draw  (v7) edge (v2);
\draw  (v7) edge (v5);
\draw  (v7) edge (v6);
\draw  (v7) edge (v4);
\draw  (v3) edge (v2);
\draw  (v3) edge (v6);
\draw  (v3) edge (v5);
\draw  (v8) edge (v9);
\draw  (v10) edge (v8);
\draw  (v9) edge (v11);
\draw  (v9) edge (v12);
\draw  (v9) edge (v13);
\draw  (v10) edge (v11);
\draw  (v10) edge (v12);
\draw  (v10) edge (v13);
\draw  (v14) edge (v8);
\draw  (v14) edge (v11);
\draw  (v14) edge (v12);
\draw  (v14) edge (v13);
\draw  (v2) edge (v15);
\draw  (v5) edge (v16);
\draw  (v6) edge (v17);
\draw  (v4) edge (v18);
\draw  (v8) edge (v15);
\draw  (v11) edge (v16);
\draw  (v12) edge (v17);
\draw  (v13) edge (v18);
\node[draw,circle] (v20) at (3,3.5) {};
\node[draw,circle] (v19) at (4,3.5) {};
\node[draw,circle] (v21) at (3.5,2.5) {};
\draw [red edge] (v15) edge[bend left=20] (v19);
\draw [red edge] (v16) edge[bend left=20] (v19);
\draw [red edge] (v17) edge[bend right=20] (v20);
\draw [red edge] (v18) edge[bend right=20] (v20);
\draw  (v20) edge (v21);
\draw  (v19) edge (v21);
\draw  (v19) edge (v21);
\draw [dashed] (2.75,3.75) rectangle (4.25,2.25);

\node at (4,2.5) {$r$};
\node at (1,9.25) {$D_A$};
\node at (6,9.25) {$D_B$};
\node at (3.5,8.5) {$S$};
\end{tikzpicture}
\caption{Lower-bound graph when \(t = 4\) }
\label{fig:lowerboundgraphTh3}
\end{figure}

The set of states is \(Q = \{{\sf True},{\sf False},\textsf{Error}\}\). The local rule of a node varies with each gadget. However, on every node, the local rule states that when a neighbor is in state \(\textsf{Error}\), the node also switches to \(\textsf{Error}\). Once in this state, it remains there forever.

Let \(t>0\) be a constant that we will fix later. The \emph{selector} is simply a set of vertices \(S = \{s_1, \dots, s_t\}\). Their local function states that they never switch their initial state unless they have a neighbor in state \(\textsf{Error}\). The \emph{collector} \(C\) is a complete binary tree rooted in a node \(r\), and its leaves are the elements of \(S\). The local rule of the internal nodes of the collector states that they never switch their initial state unless they have a neighbor in state \(\textsf{Error}\). The local rule of the root \(r\) states that when the state is different than \(\textsf{Error}\), it negates its own value, switching from \({\sf True}\) to \({\sf False}\) and vice-versa.

Observe that the selector and collector do converge to a fixed point in at most \(\lceil \log t\rceil\) time-steps if some node is in state \(\textsf{Error}\). Every configuration where no vertex is in state \(\textsf{Error}\) reaches a limit-cycle of period two.

The \emph{binary decoder gadget} is constructed from a Boolean circuit as described in Lemma~\ref{lem:binarydecoder}. In the lower-bound graph, we create two copies, \(D_A\) and \(D_B\), of the binary decoder. Each gate in the Boolean circuits is represented by a node, with links between gates represented as undirected edges. We denote the inputs and outputs of \(D_A\) as \(d_1^A, \dots, d_t^A\) and \(v_1^A, \dots, v_{2^t}^A\), respectively. Similarly, the inputs and outputs of \(D_B\) are denoted as \(d_1^B, \dots, d_t^B\) and \(v_1^B, \dots, v_{2^t}^B\). The local functions of the nodes in \(D_A\cup D_B\) mirror the Boolean functions of their corresponding gates. Each node has pointers to identify its incoming and outgoing neighbors. More precisely, let \(g\) be a gate in the binary decoder, and \(u\) be a node in \(D_A\) or \(D_B\) simulating \(g\). If \(g\) is a disjunction (respectively, conjunction, negation), the local function of \(u\) dictates that its state is the disjunction (respectively, conjunction, negation) of its incoming neighbors. If this condition is not met, node \(u\) switches to \(\textsf{Error}\).

Then, we connect \(D_A\) and \(D_B\) to the selector. For each \(i \in [t]\), we add the edges \(\{d_i^A, s_i\}\) and \(\{d_i^B, s_i\}\). The local function of \(v_i^A\) and \(v_i^B\) states that the state of the node equals that of \(s_i\). If this condition is not met, the node switches to \(\textsf{Error}\).

Finally, we modify the rules for nodes simulating the output gates. Let \({\bf A}, {\bf B}  \subseteq [2^t]\). For each \(i \notin {\bf A}\) (respectively, \(i \notin \bf B\)), the local function of \(v_i^A\) (respectively, \(v_i^B\)) dictates that, in addition to simulating the Boolean function of the corresponding output gate of the binary decoder, it switches to \(\sf{Error}\) if its state is \(\sf{True}\). 

We denote the obtained lower-bound graph as \(G\) and \(F({\bf A, B})\) as the finite-state dynamic defined over \(G\). Observe that the maximum degree of \(G\) is \(3\). We claim that when \({\bf A}\cap {\bf B} = \emptyset\), the system \((G, F({\bf A, B}))\) converges to a fixed point in at most \(6t\) time-steps. When \({\bf A}\cap {\bf B} \neq \emptyset\), there exists a configuration that does not converge. Suppose first that there is a value \(i\in {\bf A}\cap {\bf B}\). Define the following configuration over \(G\):
\begin{itemize}
\item[\(\bullet\)] The nodes in the collector gadget have all state \({\sf False}\).
\item[\(\bullet\)] For each \(j\in [t]\), the value of \(s_j\) is \({\sf True}\) if the \(j\)-th digit in the binary representation of \(i\) is \(1\), and \({\sf False}\) otherwise.
\item[\(\bullet\)] The gates simulating the inputs of \(D_A\) and \(D_B\) are assigned copies of the truth values of the selector gadgets. The remaining nodes of \(D_A\) and \(D_B\) take the Boolean value of their corresponding gate based on the evaluation of the binary decoder with the input given by the selector gadget.
\end{itemize}

Observe that by the construction of the configuration, no vertex is initially in state \(\sf{Error}\). This implies that no node in the selector or the collector will switch to \(\sf{Error}\) unless triggered by another node. Also, according to the configuration, nodes simulating internal nodes or inputs of the binary decoder do not switch to \(\sf{Error}\). In fact, by the definition of the binary decoder, only \(v_i^A\) and \(v_i^B\) are in state \(\sf{True}\). Since \(i\in {\bf A}\cap {\bf B}\), both nodes remain in their state. Consequently, no vertex of the graph switches to \(\sf{Error}\). As the collector creates a cycle of period two, we conclude that the system does not converge.

Now let us suppose that \({\bf A}\cap {\bf B} = \emptyset\), and let \(x\) be any configuration of \(G\). If any node in \(G\) is in state \({\sf Error}\), that state spreads through the network, reaching a fixed point where every node is in state \({\sf Error}\) within at most twice the depth of the binary decoder, corresponding to \(6t\) time-steps. Therefore, let us assume that no vertex is in state {\sf Error} in \(x\). Let \(i\in[2^t]\) be the index whose binary representation is encoded in the selector gadget. If the truth value of any node in \(D_A\cup D_B\) given by \(x\) does not match the truth value of the corresponding gate in the binary decoder, then that node switches to {\sf Error}. We then assume that the truth values of all the nodes in \(D_A\cup D_B\) align with those of the binary decoder. In particular, nodes \(v_i^A\) and \(v_i^B\) are {\sf True} in \(x\). Since \(i \notin {\bf A}\cap {\bf B}\), either \(v_i^A\) or \(v_i^B\) switches to {\sf Error} in the next time-step. In any case, we obtain that the dynamic reaches a fixed point in at most \(6t\) time-steps for every initial configuration.

We deduce that \((G, F({\bf A,B})) \in \convergensen(6t,3)\) if and only if \({\bf A}\) and \({\bf B}\) are yes-instances of \disj. Let \(\pi\) be a PLS for \(\convergensen(6t,3)\). We define the following two-player protocol \(\mathcal{P}\) for \disj. Given the instance \({\bf A}\), Alice nondeterministically computes the certificates that \(\pi\) would assign to all nodes in \(D_A \cup S\), and simulates the verification protocol of \(\pi\) on all nodes in \(D_A\). If every node in \(D_A\) accepts, then Alice communicates to Bob all the certificates of \(S \cup \{d_1^A, \dots, d_t^A\}\). Otherwise, she communicates one bit indicating a rejection. Bob computes the certificates of \(D_B\cup C\) and simulates the verification protocol of \(\pi\) on all nodes in \(D_B \cup S \cup C\) using the certificates for \(D_B, C\), and those for \(S\) and \(d_1^A, \dots, d_t^A\) provided by Alice. Bob accepts if the nodes on the side of Alice accepted, and also ensures that the nodes in \(D_B \cup S \cup C\) accept. The correctness of \(\mathcal{P}\) follows directly from the soundness and completeness of \(\pi\). Alice has communicated \(2t\) certificates to Bob. From Proposition~\ref{prop:disj}, we have that each certificate must be of size at least \(2^t/2t\). Defining \(k = 6t\), we conclude that for graphs of degree at most \(3\) and for a \(3\)-state dynamic, \(\convergensen(k,3)\) requires certificates of size \(2^{k/6} \cdot 6/k\).
\qed
\end{proof}

Let \(k(n):\mathbb{N}\rightarrow\mathbb{N}\) be a function. For each \(q>1\), we call \(\convergensen(k(n),q)\) the problem of deciding if a given finite-state dynamic converges in at most \(k(n)\) time-steps on \(n\)-node graphs. Observe that the lower-bound graph defined in the proof of Theorem~\ref{theo:lowerbounddegree3} has \((2^{\mathcal{O}(k(n))})\) vertices. Then, by taking \(k(n) = \Theta(\log n)\), we can obtain a linear lower-bound (up to logarithmic factors) for the certification in graphs of maximum degree at most \(3\). More precisely, we deduce the following corollary.

\begin{corollary}
There exists \(k(n) = \Theta(\log n)\) such that every Proof-Labeling Scheme for the problem \(\convergensen(k(n),3)\) requires certificates of size \(\tilde{\Omega}(n)\), even when the problem is restricted to input graphs of degree at most \(3\).
\end{corollary}

\section{Discussion}

In this paper, we study local certification of the problem \(\convergensen(t,q)\) which asks whether a finite-state dynamics with \(q\) states converges in at most \(t\) time steps. We show that in general there is an upper bound of \(b_{\text{max}} \cdot \log(\text{id}_{\text{max}}) \cdot |f_{\text{max}}|\), and we show as a lower bound that certificates of size \(\Omega(n^2/\log n)\) are required even for \(t=2\) and \(q=4\). In both cases, the size of the representation of the function \(F\) plays an important role. We present the following open questions:

\begin{enumerate}
    
    \item Which specific families of functions might exhibit a succinct representation? An interesting case study might be studying if the lower bound construction is still valid in the context of set-valued functions, such as neural networks, which depend only on the subset of states in the neighborhood.

    \item How does the difficulty of solving \(\convergensen\) compare to other decision problems commonly studied for finite-state dynamics, such as the \emph{reachability problem} (determining whether a configuration \(y\) is reachable from a starting configuration \(x\) under the dynamics) and the \emph{prediction problem} (determining if the state of a node will change after \(T\) time steps, given a node, time \(T\), and an initial configuration)?

    \item What is the certification cost for other global properties of \(F\), such as reversibility, injectivity, or nilpotency (where a function is nilpotent if there is only one possible fixed point)? How are these problems related to \(\convergensen\) and at which point are they comparable? Do they require certificates of bigger size compared to \(\convergensen\)?
\end{enumerate}


\bibliographystyle{splncs04}
\bibliography{bibilo}  

\end{document}